\documentclass{article}
\usepackage{tikz,xcolor,hyperref}
\usepackage{amsmath,amscd}
\usepackage{amssymb}
\usepackage{theorem}

\newtheorem{theorem}{Theorem}
\newtheorem{proposition}[theorem]{Proposition}
\newtheorem{corollary}[theorem]{Corollary}

{\theorembodyfont{\rmfamily}%
  \newtheorem{example}[theorem]{Example}

\newenvironment{proof}{\noindent\textit{Proof.}}
{\QED\vskip\theorempostskipamount}

\def\petitcarre{\vrule height4pt width 4pt depth0pt}
\def\QED{\relax\ifmmode\eqno{\hbox{\petitcarre}}\else{%
  \unskip\nobreak\hfil\penalty50\hskip2em\hbox{}\nobreak\hfil
  \petitcarre
  \parfillskip=0pt \finalhyphendemerits=0\par\smallskip}
  \fi}

\usepackage[
    ruled
    ,linesnumbered
    ,vlined
    ,onelanguage
    ,noend
]{algorithm2e} 

\SetKw{Return}{return}%
\SetKw{And}{and}%
\SetKwInput{KwOut}{Output}
\newcommand{\False}{\textsc{false}}

\newcommand{\True}{\textsc{true}}
\newcommand{\Break}{\textsc{break}}
\newcommand{\Z}{\mathbb{Z}}

\newcommand{\ie}{{ that is, }}

\definecolor{lime}{HTML}{A6CE39}
\DeclareRobustCommand{\orcidicon}{%
	\begin{tikzpicture}
	\draw[lime, fill=lime] (0,0)
	circle [radius=0.16]
	node[white] {{\fontfamily{qag}\selectfont \tiny ID}};
	\draw[white, fill=white] (-0.0625,0.095)
	circle [radius=0.007];
	\end{tikzpicture}
	\hspace{-2mm}
}
\foreach \x in {A, ..., Z}{%
 \expandafter\xdef\csname orcid\x\endcsname{\noexpand%
 \href{https://orcid.org/\csname orcidauthor\x\endcsname}{\noexpand\orcidicon}}
}

\title{Checking whether a word is Hamming-isometric in linear time}

\author{Marie-Pierre B\'eal\orcidA{}%
  and Maxime Crochemore\orcidB{}\\
  Univ. Gustave Eiffel, CNRS, LIGM\\ F-77454 Marne-la-Vall\'ee, France,}

\begin{document}
\maketitle
\begin{abstract}
  A finite word $f$ is Hamming-isometric if for any two words $u$ and $v$ of
  the same length avoiding $f$, $u$ can be transformed into $v$ by changing
  one by one all
the letters on which $u$ differs from $v$, in such a way that all of the
new words obtained in this process also avoid~$f$. Words which are not Hamming-isometric have been characterized as words having a border with two
mismatches. We derive from this characterization a linear-time algorithm to
check whether a word is Hamming-isometric. It is based on pattern matching algorithms
with $k$ mismatches. Lee-isometric words over a four-letter alphabet have
been characterized as words having a border with two Lee-errors. We derive from this characterization a linear-time algorithm to
check whether a word over an alphabet of size four is Lee-isometric.
\end{abstract}
{\bf Keywords:} Isometric words; Pattern matching with mismatches.
\section{Introduction}
Many parallel processing applications have communication patterns that can be
viewed as graphs called $d$-ary $n$-cubes. A $d$-ary $n$ cube is a graph
$Q_n^d$ whose nodes are the words of length $n$ over the alphabet $\Z_d= \{0,
1, \ldots, d-1\}$. Two nodes are linked if and only if they differ in exactly
one position, and the mismatch is given by two symbols $a$ and $b$ that
verify $a = b \pm 1 \mod d$. In order to obtain some variants of hypercubes
for which the number of vertices increases slower than in a hypercube, Hsu
\cite{Hsu1993} introduced Fibonacci cubes in which nodes are on a binary
alphabet and avoid the factor $11$. The notion of $d$-ary $n$-cubes has
subsequently be extended to define the generalized Fibonacci cube
\cite{IlicKlavzarRho2012, Klavzar2013, WeiZhang2016} ; it is the subgraph
$Q_n^2(f)$ of a $2$-ary $n$-cube whose nodes avoid some factor $f$. In this
framework, a binary word $f$ is said to be Lee-isometric when, for any $n \geq
1$, $Q_n^2(f)$ can be isometrically embedded into $Q_n^2$, \ie the
distance between two words $u$ and $v$ vertices of $Q_n^2(f)$ is the
same in $Q_n^2(f)$ and in $Q_n^2$.

On a binary alphabet, the definition of a Lee-isometric word can be
equivalently given by ignoring hypercubes and adopting a point of view closer
to combinatorics on words. A binary word $f$ is $n$-Hamming-isometric if for
any pair of words $u$ and $v$ of length $n$ avoiding $f$, $u$ can be
transformed into $v$ by exchanging one by one the bits on which they differ
meanwhile generating only words avoiding $f$. The word $f$ is
Hamming-isometric if it is $n$-Hamming-isometric for all $n$. The structure
of binary non-Hamming-isometric words has been characterized in
\cite{KlavzarShpectorov2012, Wei2017, WeiYangZhu2019} and extended to general
alphabets in \cite{AnselmoEtAl2021}. In particular, a binary word is
Hamming-isometric if and only if it is Lee-isometric. A word is not
Hamming-isometric if and only if it has a $2$-error border, that is if it has
a suffix that mismatches with the prefix of the same length in exactly two
positions. In \cite{KlavzarShpectorov2012, Wei2017, WeiYangZhu2019}
and \cite{AnselmoEtAl2021}, $2$-error border are called $2$-error overlap.

In the case of an alphabet of size $4$, non-Lee-isometric words have been
characterized in \cite{AnselmoEtAl2021} as words having a suffix and a prefix
of the same length which are at distance $2$ according to the Lee distance.

Binary Hamming-isometric words have also been considered in the
two-dimensional setting, and non-Hamming-isometric pictures are investigated
in \cite{AnselmoEtAl2020}, where they are called bad pictures.

In this paper we study the algorithmic complexity of checking whether a word
is not Hamming-isometric. Our approach is based on the characterization of $2$-error borders of
such words. The naive algorithm runs clearly in quadratic time. We show that
known algorithms for matching patterns with mismatches can be used to solve
this problem efficiently. Pattern matching with $k$ mismatches can be solved
by algorithms running in time $O(nk)$ (see \cite{LandauVishkin1985} and
\cite{GalilGiancarlo1986}). These algorithms are mostly based on a technique
called the Kangaroo method. This method computes the Hamming distance for
every alignment in time $O(k)$ by “jumping” from one error to the next error.
A faster algorithm for pattern matching with $k$ mismatches runs in $O(n
\sqrt{k \log k})$ \cite{AmirLewensteinPorat2004}. A simpler version of this
algorithm is given in \cite{NicolaeRajasekaran2017}.

We show two methods to check whether a word is not Hamming-isometric. The
first one uses the Kangaroo method which allows to derive an algorithm
running in time $O(kn)$ and using $O(n)$ space to check whether a word of
length $n$ has a $k$-error border. The method has a preprocessing of linear
time and space for computing the suffix tree of the word and to enhance it in
order to answer lowest common ancestor queries in constant time. This overall
leads to a linear-time and linear-space algorithm to check whether a word is
not Hamming-isometric, and hence also to check whether a binary word is
Lee-isometric. The second method uses the computation of a $k$-prefix table
that gives, for some word $u$ and every position on $u$, the length of the
longest proper factor of $u$ at this position that matches its prefix of the
same length with at most $k$ differences \cite{BartonEtAl2014}. The
computation of this $k$-prefix table is done in time $O(kn)$ using $O(n)$
space.

We also use the Kangaroo method to derive an algorithm running in time $O(kn)$ and
using $O(n)$ space on a constant size alphabet to check whether a word of size $n$ has a $k$-Lee-error
border and thus check in linear time whether a word over an alphabet
of size 4 is Lee-isometric.

\section{Definitions and background}
Let $A$ be a finite alphabet. A word $u$ in $A^\star$ is a finite sequence
$u[0]u[1] \cdots u[n-1]$ of letters in $A$, where $n$ is the length of $u$
and $u[i]$ are its letters. The suffix of index $i$ on $u$, denoted by $u_i$,
is the word $u[i] \cdots u[n-1]$ of length $n-i$. A suffix (or prefix) of it
is \emph{proper} if it is distinct from $u$ itself.

Let $k$ be a non-negative integer. We say that a word $u$ has a
$k$-\emph{error border} if $u$ has a proper suffix $s$ that matches its prefix of
the same length with exactly $k$ differences.
In other words, the Hamming distance between the suffix and the prefix is
$k$.

\begin{example}
  The word $1010011$ has a $2$-error border. Indeed, it has the prefix
  $101$ and the suffix $011$ and the Hamming distance between $101$
  and $011$ is $2$.
 \end{example}

Let $f$ be a finite word and $n$ be a positive integer. Then a word $u$ is
called $f$-\emph{free} if it does not contain $f$ as a factor, and $f$ is
called $n$-\emph{Hamming-isometric} if for every $f$-free words $u$ and $v$
of length $n$, the following holds: $u$ can be transformed into $v$ by
changing one by one all the letters on which $u$ differs from $v$, in such a
way that all of the new words obtained during this process are also $f$-free.
Such a transformation is called an $f$-\emph{free transformation} from $u$ to
$v$. Eventually, a word $f$ is said to be \emph{Hamming-isometric} if it is
$n$-\emph{Hamming-isometric} for every positive integer $n$.

The $d$-ary $n$-\emph{cube}, denoted by $Q_n^d$, is the graph whose vertices
are the words of length $n$ over the alphabet $\mathbb{Z}_d = \{0, 1, \ldots,
d-1\}$, and for which any two words $u$ and $v$ are adjacent if and only if
$u$ and $v$ differ by one unit at exactly one position, say $i$, that is,
$u[i] = v[i] \pm 1 \mod d$. The $d$-ary $n$-\emph{cube
avoiding $f$}, where $f$ is a word over the alphabet $\mathbb{Z}_d$ is the
graph $Q_n^d(f)$ obtained from $Q_n^d$ by deleting the vertices containing
$f$ as a factor \cite{AnselmoEtAl2021}.

A word $f$ over $\mathbb{Z}_d$ is said to be \emph{Lee-isometric} if for all $n \geq 1$,
$Q_n^d(f)$ is an isometric subgraph of $Q_n^d$.

\begin{example}
  The word $0301$ on the alphabet $\mathbb{Z}_4=\{0, 1, 2, 3\}$ is
  non-Lee-isometric. Indeed, the words $u=030001$ and $v = 030201$, which do not
  contain the factor $0301$, are at distance 2 but there is no path of
  length $2$ from $u$ to $v$ in $Q_6^4(0301)$ since any path of length $1$ changing the
  symbol of index $3$ of $u$ goes from $u$ to $030101$ or to $030301$
  and these two words both have the word $0301$ as factor.
  \end{example}

It is shown in \cite{AnselmoEtAl2021} that non-Hamming-isometric and non-Lee-isometric words
coincide for words on an alphabet of size at most three. But this property is
no more true for larger alphabets.

Hamming-isometric words have the following characterization obtained in
\cite{KlavzarShpectorov2012, Wei2017} for binary alphabets and in
\cite{AnselmoEtAl2021} for general alphabets.

\begin{proposition} \label{proposition.characterization}
  A word is not Hamming-isometric if and only if it has a 2-error border.
\end{proposition}
\begin{example}
  For instance the words $11$, $1^n$ for $n \geq 1$ are Hamming-isometric. The word
  $1010011$ is not Hamming-isometric.
\end{example}

\section{Algorithms for checking whether a word is Hamming-isometric}

In this section, we use the characterization of non-Hamming-isometric words
in terms of 2-error border (Proposition~\ref{proposition.characterization})
and assume that the alphabet $A$ has a constant size. Observe that a
quadratic-time naive algorithm can be obtained to check whether a word is
non-Hamming-isometric by computing the Hamming distance between each
suffix of index $i$ and the prefix of the same length. We show that checking
if a word of length $n$ has a $k$-error border can be done in time $O(kn)$
and space $O(n)$.

\begin{proposition}\label{prop2}
It can be checked in time $O(kn)$ and space $O(n)$ whether a word of length
$n$ has a $k$-error border.
\end{proposition}

\begin{proof}
We give two algorithms for solving this problem. The first one is based on a
technique called the Kangaroo method used for pattern matching with $k$
mismatches in $O(nk)$ time (see \cite{LandauVishkin1985},
\cite{GalilGiancarlo1986} and \cite{NavarroRaffinot2002}). These algorithms
compute the Hamming distance for every alignment in $O(k)$ time by “jumping”
from one error to the next. We use the Kangaroo method to check for each
index $i$ on a word $u$ of length $n$ whether it has a $k$-error border of
length $n-i$ in time $O(k)$.

To do so, we first compute in time and space $O(n)$ the suffix tree of
the word $u$.
The suffix tree is a compacted trie containing all the suffixes of $u$ by
their keys and positions on $u$ as their values 
\cite{CrochemoreHancartLecroq2007}, \cite{ApostolicoEtAl2016}. The tree has a
linear number of nodes and edges, each edge containing a pair of integers
identifying a factor of $u$, e.g. (position, length), hence the linear space
complexity. Suffix arrays can also be used for this problem. They contain
essentially the starting positions of suffixes of $u$ sorted in lexicographic
order.

To get the overall running time, we need to answer Lowest Common Ancestor
(LCA) queries in constant time \cite{HarelT84}, \cite{SchieberV88}.
LCA queries give us the longest common prefix between two suffixes of $u$,
essentially telling us where the first mismatch appears between a suffix of
$u$ and its prefix of the same length. This can be performed by first
constructing a Longest Common Prefix (LCP) array. The LCP array stores the
length of the longest common prefix between two consecutive suffixes in the
suffix array (lexicographic consecutive suffixes). This array can also be
constructed in linear time. To compute the length of the longest prefix
common to any two suffixes in the suffix tree (instead of consecutive
suffixes), we need to use some range minimum query data structure.

Thus, we assume that our suffix tree is enhanced to answer LCA queries in
constant time. This can be done in linear time and space. We denote by
$\text{LCA}(i,j)$ the query that returns in $O(1)$ time the length of the
common prefix between the suffix $u_i$ and the suffix $u_j$ of $u$.

For every index $i$, we try to find if the suffix of index $i$ of the
word $u$ has
$k$ mismatches with its prefix of the same length. We first compute
$\text{LCA}(0,i)$. Let this length be $\ell_0$. We skip the mismatching
character in $u_0$ and $u_i$ and try to find
$\text{LCA}(\ell_0+1,i+\ell_0+1)$. We repeat this to obtain $k$ mismatches
between $u_i$ and $u[0] \cdots u[n-i-1]$ or fail to obtain this condition.

The pseudo code of the technique is given in Algorithm 1. We maintain a
variable $\ell$ which gives, after the line 4 of Algorithm 1, the index 
of the current mismatch between $u_i$ and $u_0$. A variable $d$ contains the
current Hamming distance between $u[i] \cdots u[i + \ell -1]$ and $u[0]
\cdots u[\ell -1]$. It is increased by 1 at the line 8 since a mismatch has been
found.

Since there are at most $O(k)$ LCA queries for each index $i$, this can be
done in $O(k)$ time. The overall time complexity is thus $O(kn)$ and the
space complexity is $O(n)$.

We now show a second method to check whether a word of length $n$ has a
$k$-error border. We use the computation of a $k$-\emph{prefix table} as done
in \cite{BartonEtAl2014}. For each position $i$, we compute a table $\pi_k$
for which $\pi_k(i)$ is the length $\ell$ of the longest word $u[i] \ldots
u[i+ \ell-1]$ such that the Hamming distance between $u[0] \ldots u[\ell -1]$
and $u[i] \ldots u[i + \ell-1]$ is at most $k$ and $u[0] \ldots u[\ell -1]$
is proper prefix of $u$.

This computation can be done in time $O(kn)$ and space $O(n)$ (see
\cite[Theorem 5]{BartonEtAl2014}). It needs the computations of the
prefix array of $u$ and the longest common prefix array preprocessed
for range minimum queries. The longest common prefix array gives for
each index $r$ the length of the longest common prefix of the $r$th
suffix and the $(r-1)$th suffix in lexicographic order.

The existence of a $k$-error border is then obtained as follows. For $k \geq
1$, a word $u$ has a $k$-error border if and only if there is a position $i$,
$1 \leq i < n$, for which $\pi_k[i] = n-i$ and $\pi_{k-1}[i] < n-i$. Indeed
such a position $i$ exists if and only if there is a proper suffix $u[i]
\ldots u[n-1]$ of $u$ whose Hamming distance with $u[0] \ldots u[n-i-1]$ is
exactly $k$. The existence of a $k$-error border is thus obtained with
Algorithm \ref{algorithm.2} which is in $O(n)$ time. The overall time
complexity is again in $O(kn)$ and the space complexity is $O(n)$.
\end{proof}

\begin{algorithm} \label{algorithm.1}
\KwIn{A non empty word $u$ of length $n$, a non-negative integer $k$}
\KwOut{true if $u$ has a $k$-error border}
\BlankLine
\For{$i\leftarrow 1$ \KwTo $n-1$}{
  $(\ell, d) \leftarrow (0,0)$\;
  \While{$d \leq k$}{
          $\ell \leftarrow$ $\ell +$ LCA$(\ell, i+\ell)$\;
          \If{$d = k$ \mbox{and} $\ell = n-i$ }{
              \Return \True;
            }
            \If{$d < k$ \mbox{and} $\ell < n-i$}{
               $(\ell, d) \leftarrow (\ell+1, d+1)$\;
             }
             \Else{
               \Break
              }
        }
    }
\Return \False;
\caption{{\sc Word with a $k$-error border}($u$)}
\end{algorithm}

\begin{example}
  Let $u = 101011$. Let us check with Algorithm \ref{algorithm.1} whether $u$ has a
  $2$-error border.
  For $i=1$, at the first step of the loop of the line 3 we obtain
  at the line 4 $\ell = \mbox{LCA}(0,1) = 0$; we set $\ell$ to $1$ (the
  jump) and $d$ to $1$ at the line 8. At the second step of the loop of the line 3 we obtain
  at the line 4 $\ell = \ell + \mbox{LCA}(1,2) = 1$; we set $\ell$ to $2$ and
  $d$ to $2$ at the line 8. At the third step of the loop of the line
  3, we obtain at the line 4
  $\ell = \ell + \mbox{LCA}(2,3) = 2$ and break at the line 10.
  For $i=2$ the loop of the line 3 fails to return true.
  For $i=3$,  at the first step of the loop of the line 3 we obtain
  at the line 4 $\ell = \mbox{LCA}(0,3) = 0$; we set $\ell$ to $1$ and $d$
  to $1$ at the line 8. At the second step of the loop of the line 3 we obtain
  at the line 4 $\ell = \ell + \mbox{LCA}(1,4) = 1$;  we set $\ell$ to $2$
  and $d$ to $2$ at the line 8. At the third step of the loop of the
  line 3, we obtain the at line 4
  $\ell = \ell + \mbox{LCA}(2,5) = 3$ and, since $d=2$ and $\ell = n-i = 3$,
  the algorithm returns true at the line 6. The algorithm has thus detected
  the $2$-error border of length $3$.
  \end{example}

\begin{algorithm} \label{algorithm.2}
\KwIn{A non empty word $u$ of length $n$, a non-negative integer $k$,
  the $k$-prefix table $\pi_k$ and the $(k-1)$-prefix table $\pi_{k-1}$}
\KwOut{true if $u$ has a $k$-error border}
\BlankLine
\For{$i\leftarrow 1$ \KwTo $n-1$}{
  \If{$\pi_k[i] = n-i$ \mbox{and} $\pi_{k-1}[i] < n-i$}{
    \Return \True;
  }
}
\Return \False;
\caption{{\sc Word with a $k$-error border}($u$)}
\end{algorithm}

The following corollary follows then directly from Proposition
\ref{proposition.characterization} and the analysis of Algorithm
\ref{algorithm.1} in Proposition~\ref{prop2}.

 \begin{corollary}
   It can be checked in linear time and space whether a word is Hamming-isometric.
\end{corollary}

\section{Algorithm for checking whether a word over an alphabet of size
  $4$ is Lee-isometric}

A combinatorial characterization of Lee-isometric words over an alphabet
of size $4$ has been obtained in
\cite{AnselmoEtAl2021}. It uses the notion of Lee distance
which is defined as follows. The \emph{Lee distance}, denoted by $d_L$,
between two letters of the alphabet $\mathbb{Z}_d=\{0, 1, \ldots,
d-1\}$ is
\[
  d_L(a,b) = \min(|a - b|, d-|a - b|).
\]
The Lee distance between two words $u$ and $v$ of length $n$ over
$\mathbb{Z}_d$ is
\[
  d_L(u,v) = \sum_{i = 0}^{n-1} d_L(u[i],v[i]).
\]
A word has a $k$-\emph{Lee-error border} if it has a suffix $u$ and a
prefix $v$ of same length satisfying $d_L(u,v)=k$.

For words over $\mathbb{Z}_4$, the Lee-isometric words are characterized
as follows in \cite{AnselmoEtAl2021}.

\begin{proposition} \label{proposition.characterization2}
  A word over a $4$-letter alphabet is non-Lee-isometric if and only if it has a 2-Lee-error border.
\end{proposition}

In this section we show that checking if a word of length $n$ has a
$k$-Lee-error border can be done in time $O(kn)$ and space $O(n)$.
The algorithm is Algorithm \ref{algorithm.3}. 

\begin{algorithm} \label{algorithm.3}
\KwIn{A non empty word $u$ of length $n$, a non-negative integer $k$}
\KwOut{true if $u$ has a $k$-Lee-error border}
\BlankLine
\For{$i\leftarrow 1$ \KwTo $n-1$}{
  $(\ell,d) \leftarrow (0,0)$\;
  \While{$d \leq k$}{
          $\ell \leftarrow$ $\ell +$ LCA$(\ell, i+\ell)$\;
          \If{$d = k$ \mbox{and} $\ell = n-i$ }{
              \Return \True;
            }
            \If{$d = k$ \mbox{and} $\ell < n-i$}{
               \Break
             }
             \If{$d < k$ \mbox{and} $\ell = n-i$}{
               \Break
             }
             $d \leftarrow d+ d_L(u[\ell],u[i+\ell])$\;
             $\ell \leftarrow \ell+ 1$\;
        }
    }
\Return \False;
\caption{{\sc Word with a $k$-Lee-error border}($u$)}
\end{algorithm}

\begin{proposition}\label{prop3}
It can be checked in time $O(kn)$ and space $O(n)$ whether a word of length
$n$ has a $k$-Lee-error border.
\end{proposition}

\begin{proof}
The algorithm is almost the same as Algorithm \ref{algorithm.1} and the proof is
similar to the proof of Proposition \ref{prop2}. Therefore we only
discuss the differences.

For every index $i$, we try to find if the suffix of index $i$ is at Lee
distance $k$ from its prefix of the same length.

The pseudo code of the technique is given in Algorithm \ref{algorithm.3}. A
variable $d$ contains, after the line 11, the current Lee distance between $u[i]
\cdots u[i + \ell]$ and $u[0] \cdots u[\ell]$.

The difference with Algorithm \ref{algorithm.1} appears when there is
mismatch between the suffix $u_i$ and the prefix $u[0] \cdots u[n-i-1]$ at
positions $\ell$ on $u_i$ and $i + \ell$ on $u[0] \cdots u[n-i-1]$. The
current Lee distance between $u[0] \cdots u[\ell]$ and $u[i] \cdots
u[i+\ell]$ is augmented this time by the value of $d_L(u[\ell],
u[i+\ell])$.

The case $d < k$ and $\ell = n-i$ of the line 9 of Algorithm
\ref{algorithm.3} corresponds to the case where the suffix of $u$ at
position $i$ is a $d$-Lee-error border with $d < k$ and is thus not a
solution. The algorithms continues then to check the position $i+1$.

Since there are at most $O(k)$ LCA queries for each index $i$, this can be
done in $O(k)$ time. The overall time complexity is thus $O(kn)$ and the
space complexity is $O(n)$.

\end{proof}

The following corollary follows then directly from Proposition
\ref{proposition.characterization2} and the analysis of Algorithm \ref{algorithm.2} in
Proposition~\ref{prop3}.

 \begin{corollary}
   It can be checked in linear time and space whether a word over an
   alphabet of size $4$ is Lee-isometric.
\end{corollary}

\section{Acknowledgment}
We thank Marcella Anselmo for helpful comments.

\bibliographystyle{plain}
\bibliography{border.bib}

\begin{thebibliography}{10}

\bibitem{AmirLewensteinPorat2004}
Amihood Amir, Moshe Lewenstein, and Ely Porat.
\newblock Faster algorithms for string matching with $k$ mismatches.
\newblock {\em J. Algorithms}, 50(2):257--275, 2004.

\bibitem{AnselmoEtAl2021}
Marcella Anselmo, Manuela Flores, and Maria Madonia.
\newblock Quaternary n-cubes and isometric words.
\newblock In {\em Combinatorics on Words - 13th International Conference,
  {WORDS} 2021, Rouen, France, September 13-17, 2021, Proceedings}, volume
  12847 of {\em Lecture Notes in Computer Science}, pages 27--39. Springer,
  2021.

\bibitem{AnselmoEtAl2020}
Marcella Anselmo, Dora Giammarresi, Maria Madonia, and Carla Selmi.
\newblock Bad pictures: Some structural properties related to overlaps.
\newblock In Galina Jir{\'{a}}skov{\'{a}} and Giovanni Pighizzini, editors,
  {\em Descriptional Complexity of Formal Systems - 22nd International
  Conference, {DCFS} 2020, Vienna, Austria, August 24-26, 2020, Proceedings},
  volume 12442 of {\em Lecture Notes in Computer Science}, pages 13--25.
  Springer, 2020.

\bibitem{ApostolicoEtAl2016}
Alberto Apostolico, Maxime Crochemore, Martin Farach{-}Colton, Zvi Galil, and
  S.~Muthukrishnan.
\newblock 40 years of suffix trees.
\newblock {\em Commun. {ACM}}, 59(4):66--73, 2016.

\bibitem{BartonEtAl2014}
Carl Barton, Costas~S. Iliopoulos, Solon~P. Pissis, and William~F. Smyth.
\newblock Fast and simple computations using prefix tables under {H}amming and
  edit distance.
\newblock In Jan Kratochv{\'{\i}}l, Mirka Miller, and Dalibor Froncek, editors,
  {\em Combinatorial Algorithms - 25th International Workshop, {IWOCA} 2014,
  Duluth, MN, USA, October 15-17, 2014, Revised Selected Papers}, volume 8986
  of {\em Lecture Notes in Computer Science}, pages 49--61. Springer, 2014.

\bibitem{CrochemoreHancartLecroq2007}
Maxime Crochemore, Christophe Hancart, and Thierry Lecroq.
\newblock {\em Algorithms on {S}strings}.
\newblock Cambridge University Press, 2007.

\bibitem{GalilGiancarlo1986}
Z~Galil and R~Giancarlo.
\newblock Improved string matching with $k$ mismatches.
\newblock {\em SIGACT News}, 17(4):52–54, March 1986.

\bibitem{HarelT84}
Dov Harel and Robert~Endre Tarjan.
\newblock Fast algorithms for finding nearest common ancestors.
\newblock {\em {SIAM} J. Comput.}, 13(2):338--355, 1984.

\bibitem{Hsu1993}
W.-J. Hsu.
\newblock Fibonacci cubes-a new interconnection topology.
\newblock {\em IEEE Transactions on Parallel and Distributed Systems},
  4(1):3--12, 1993.

\bibitem{IlicKlavzarRho2012}
Aleksandar Ilic, Sandi Klavzar, and Yoomi Rho.
\newblock Generalized {F}ibonacci cubes.
\newblock {\em Discret. Math.}, 312(1):2--11, 2012.

\bibitem{Klavzar2013}
Sandi Klavzar.
\newblock Structure of {F}ibonacci cubes: a survey.
\newblock {\em J. Comb. Optim.}, 25(4):505--522, 2013.

\bibitem{KlavzarShpectorov2012}
Sandi Klavzar and Sergey~V. Shpectorov.
\newblock Asymptotic number of isometric generalized {F}ibonacci cubes.
\newblock {\em Eur. J. Comb.}, 33(2):220--226, 2012.

\bibitem{LandauVishkin1985}
Gad~M. Landau and Uzi Vishkin.
\newblock Efficient string matching in the presence of errors.
\newblock In {\em 26th Annual Symposium on Foundations of Computer Science,
  Portland, Oregon, USA, 21-23 October 1985}, pages 126--136. {IEEE} Computer
  Society, 1985.

\bibitem{NavarroRaffinot2002}
Gonzalo Navarro and Mathieu Raffinot.
\newblock {\em Flexible {P}attern {M}atching in {S}trings - practical on-line
  search algorithms for texts and biological sequences}.
\newblock Cambridge University Press, 2002.

\bibitem{NicolaeRajasekaran2017}
Marius Nicolae and Sanguthevar Rajasekaran.
\newblock On pattern matching with $k$ mismatches and few don't cares.
\newblock {\em Inf. Process. Lett.}, 118:78--82, 2017.

\bibitem{SchieberV88}
Baruch Schieber and Uzi Vishkin.
\newblock On finding lowest common ancestors: Simplification and
  parallelization.
\newblock {\em {SIAM} J. Comput.}, 17(6):1253--1262, 1988.

\bibitem{Wei2017}
Jianxin Wei.
\newblock The structures of bad words.
\newblock {\em Eur. J. Comb.}, 59:204--214, 2017.

\bibitem{WeiYangZhu2019}
Jianxin Wei, Yujun Yang, and Xuena Zhu.
\newblock A characterization of non-isometric binary words.
\newblock {\em Eur. J. Comb.}, 78:121--133, 2019.

\bibitem{WeiZhang2016}
Jianxin Wei and Heping Zhang.
\newblock Proofs of two conjectures on generalized {F}ibonacci cubes.
\newblock {\em Eur. J. Comb.}, 51:419--432, 2016.

\end{thebibliography}
\end{document}